\documentclass[letterpaper, 10pt, conference]{ieeeconf}
\IEEEoverridecommandlockouts                              
\usepackage{graphics} 
\usepackage[pdftex]{graphicx}
\usepackage{epsfig} 
\usepackage{mathptmx} 
\usepackage{times} 
\usepackage{amsmath} 
\usepackage{amssymb}  
\usepackage{subfigure}
\usepackage{verbatim}
\usepackage{cite}
\usepackage{epstopdf}
\usepackage{color}
\usepackage{relsize}
\usepackage{nopageno}
\newtheorem{theorem}{Theorem}

\newtheorem{assumption}{Assumption}
\usepackage{color}

\newtheorem{corollary}{Corollary}

\newtheorem{lemma}{Lemma}

{ 

\newtheorem{example}{Example}
}

\newcounter{prob1}
\newcounter{prob2}
\newcounter{prob3}
\newcounter{prob4}
\newcounter{prob5}
\newcounter{prob6}
\newcounter{prob7}

\title{\LARGE \bf Control approach to  computing the feedback capacity for stationary finite dimensional Gaussian channels}

\author{ \parbox{5 in}{\centering Chong Li* and Nicola Elia
         \thanks{This work was supported by NSF under grant number ECS-0901846. Emails:chongl@qti.qualcomm.com; nelia@iastate.edu}\\
         *Qualcomm Research, NJ \\
         Department of Electrical and Computer Engineering, Iowa State University, IA\\}}

\begin{document}

\maketitle \thispagestyle{empty} \pagestyle{plain}
\begin{abstract}
We firstly extend the interpretation of feedback communication over stationary finite dimensional Gaussian channels as feedback control systems by showing that, the problem of finding stabilizing feedback controllers with maximal reliable transmission rate over Youla parameters coincides with the problem of finding strictly causal filters to achieve feedback capacity recently derived in \cite{Kim10}. The aforementioned interpretation provides an approach to construct deterministic feedback coding schemes (with double exponential decaying error probability).
 We next propose an asymptotic capacity-achieving upper bounds, which can be numerically evaluated by solving finite dimensional dual optimizations.
From the filters that achieve upper bounds, we derive feasible filters which lead to a sequence of lower bounds. Thus,  from the lower bound filters we obtain  communication systems that achieve the lower bound rate. Extensive examples show the sequence of lower bounds is asymptotic capacity-achieving as well.
\end{abstract}


\setcounter{prob1}{1}
\setcounter{prob2}{2}
\setcounter{prob3}{3}
\setcounter{prob4}{4}
\setcounter{prob5}{5}
\setcounter{prob6}{6}
\setcounter{prob7}{7}

\section{Introduction}

\indent We consider a discrete-time Gaussian channel with noiseless feedback.
The additive Gaussian channel is modeled as
\begin{equation}\label{model: forward channel}
Y_i=U_i+W_i, \qquad i=1,2,\cdots
\end{equation}
where the Gaussian noise $\lbrace W_i \rbrace_{i=1}^{\infty}$ is assumed to be stationary with power spectrum density $\mathbb{S}_{w}(e^{i\theta})>0$ for $\forall \theta\in [-\pi,\pi)$. Unless the contrary is explicitly stated, ``stationary'' without specification refers to stationary in wide sense. Let $\mathcal{RH}_2$ be the set of stable, strictly proper \footnote{As will be seen later, strictly proper introduces one step delay in the feedback processing of the information from $y_i$
} rational filters in Hardy space $\mathcal{H}_2$.
\begin{assumption}\label{LTIFD.ass}
In this paper, noise $W$ is assumed as the output of a finite-dimensional linear time invariance (LTI) stable system $\mathbb{H} \in \mathcal{RH}_2$, not necessarily minimum phase,  driven by white Gaussian noise with zero mean and unit variance.
The power spectral density of $W$ is given by  ${\mathbb S}_w(e^{j\theta})= |\mathbb{H}(e^{j\theta})|^2$ (i.e. canonical spectral factorization).
\end{assumption}

Note that any stationary process can be approximated with arbitrary accuracy by this LTI filtering model and this approximation is very ''efficient'', as it corresponds to the rational approximation of the spectral density \cite{book.stat.linear.sys}.

For a code of rate $R$, we specify a  $(n,2^{nR})$ channel code as follows. $M$ is a uniformly distributed message index where $M\in\lbrace 1,2,3,\cdots,2^{nR}\rbrace$. There exists an encoding process $U_i(M,Y^{i-1})$, where $Y^{i-1}=\lbrace Y_0, Y_1,\cdots, Y_{i-1}\rbrace)$, for $i = 1,2,\cdots,n$ and $U_1(M, Y^0) = U_1(M)$ with average transmit power constraint and a decoding function g:$Y^n\rightarrow \lbrace 1,2,\cdots,2^{nR}\rbrace$ with an error probability satisfying
$P_e^{(n)}=\frac{1}{2^{nR}}\sum_{M=1}^{2^{nR}} p(M\neq g(y^n)|M)\leq \epsilon_n$,
where $\lim_{n\rightarrow\infty}\epsilon_n=0$.  This coding process indicates that the channel input $U_i$ is determined by the message index $M$ and previous channel output $Y^{i-1}$. The objective of communication is to deliver $M$ to the receiver at highest code rate with arbitrarily small error probability. The feedback capacity $C_{fb}$ is defined as the supremium of all achievable rates $R$.
\indent 
%

As it is shown in \cite{Kim10}, if $\mathbb{S}_w$ has a canonical spectral factorization, the feedback capacity can be characterized by 
\begin{equation}
\begin{split}
C_{fb}=&\max_{\mathbb{Q}}\frac{1}{2\pi}\int_{-\pi}^{\pi}\log |1+\mathbb{Q}(e^{i\theta})|d\theta,\\
s.t. \quad  &\frac{1}{2\pi}\int_{-\pi}^{\pi}|\mathbb{Q}(e^{i\theta})|^2\mathbb{S}_w(e^{i\theta})d\theta\leq P,\\
& \mathbb{Q}(e^{i\theta}) \text{is strictly causal}, \;{ \mathbb{Q}\in \mathcal{RH}_2} .\\
\end{split}
\label{capacity_short01}
\end{equation}

While the above characterization is elegant, it is infinite dimensional, and  except for the \textit{first-order auto-regressive moving average} (ARMA) noise, finding the feedback capacity, either analytically or numerically, remains open.\\
\indent In this paper, we revisit and extend the interpretation of  feedback communication over Gaussian channels as feedback control problems, \cite{unified.theory}. In particular, we highlight the central role of \textit{Youla parameterization over all stabilizing controllers} in connecting these two theories and 
show that  the characterization of the maximum-rate over all stabilizing controllers and the feedback capacity over all coding schemes (\ref{capacity_short01}) coincide.
Moreover, our result provides the explicit (sub-)optimal communication scheme (i.e. encoder and decoder) directly from $\mathbb{Q}$. It is worth noting that \cite{Kim10} (Theorem 6.1 and Lemma 6.1) has shown an $k$-dimensional generalization of the Schalkwijk-Kailath coding scheme achieves the feedback capacity for any ARMA noise spectrum of order $k$. Alternatively, we herein provide a feedback coding scheme from $\mathbb{Q}$ by leveraging control-oriented derivations, which can be directly constructed and implemented.
We next provide an alternative characterization of the feedback capacity, from which an asymptotic capacity-achieving sequence of upper bounds is derived and can be numerically evaluated by solving finite dimensional optimizations. Furthermore, a sequence of lower bounds on the feedback capacity are obtained by constructing specific deterministic coding schemes with double exponential decaying error probability, which reveal a direct connection with Youla parameter besides being an generalized Schalkwijk-Kailath scheme studied by Elia \cite{Elia2004}, Kim\cite{Kim10}, Liu-Elia\cite{Liu_CIS}, Shayevitz-Feder\cite{posterior.matching} and others. The archived lower bound and the upper bounds provide a way to evaluate how close is the scheme to the feedback capacity.  From extensive examples, the sequence of lower bounds converges arbitrarily close to the capacity, yielding an asymptotic optimal feedback coding scheme.

\subsection{Related Work}
We review the literature along two avenues of information theory and feedback control theory. As a complete survey is vast and most of them are out of the scope of our discussion, we herein list most relevant results to this paper. In the field of information theory, the investigation on feedback Gaussian capacity has been experiencing a decade journey. \cite{Elias1956} and its sequel \cite{Elias1967} are recognized as the first work on feedback Gaussian channels by proposing feedback coding schemes. \cite{Schalkwijk66} \cite{Schalkwijk66_2} developed an elegant linear feedback coding scheme of achieving the capacity of additive white Gaussian noise (AWGN) channel with noiseless feedback. Thereafter, several work by Butman \cite{Butman69}, \cite{Butman76}, Tiernan\cite{Tiernan74}\cite{Tiernan76}, Wolfowitz\cite{Wolfowitz75} and Ozarow\cite{Ozarow_random90}\cite{Ozarow_upper90} extended this notable result to ARMA Gaussian channels, with objective to find channel capacity and optimal feedback codes. As a consequence, many interesting upper and lower bounds were obtained. Based on the insight/results from aforementioned literature, \cite{cover89} made a major breakthrough on characterizing the n-block capacity of arbitrary feedback Gaussian channels by using asymptotic equipartition property (AEP) theorem. It was also shown that feedback capacity for arbitrary Gaussian channels cannot be increased by factor two or half bit. This n-block capacity was extended to the case of feedback Gaussian channels with noisy feedback where capacity bounds and other interesting results were obtained\cite{chong11_ISIT}\cite{Chong11_allerton_upperbound,Chong12_allerton_sideInfo,chong.thesis}. As hinted by this n-block capacity characterization, \cite{Kim10} developed a variational characterization on the capacity of stationary feedback Gaussian channels, which is an infinite dimensional optimization problem. For first-order ARMA noise, this variational characterization yields a closed-form solution on the capacity and shows the optimality of the Schalkwijk-Kailath scheme.

In the field of feedback control theory, many control-based technical tools have been utilized to attack the problem of finding feedback Gaussian channel capacity and capacity-achieving codes. \cite{Elia2004} proposed the derivation of  feedback communication schemes based on a feedback control method.
These results were  obtained from considering the problem of stabilization of a given unstable plant over a Gaussian communication channel.  The communication rate (in the sense of Shannon) over the channel was connected to the degree on instability of the plant. The minimal transmission power for a given unstable plant was obtained by solving the classical ${\cal H}_2$ (or Linear Quadratic Gaussian) problem. However, plants with the same degree of instability may require different transmission power to be stabilized.  \cite{Elia2004} provided the plants that can be stabilized most efficiently, i.e. with the least transmission power for a given degree of instability for special case channels. This approach provides a method of finding feedback coding scheme for Gaussian channels. The approach has been further extended  to various channels \cite{Liu04_ISIT}, connected to the classical Linear Quadratic Gaussian (LQG) control problem \cite{Franceschetti_contrl_comm_fd}. Finally, \cite{Liu_CIS} extended the convergence of the fundamental limitations of control and communication to include the limitations of estimation. In light of this unified framework, a set of achievable rates of feedback Gaussian channels were obtained by constructing specific feedback coding schemes via control-oriented approaches.
\cite{Yang_feedbackCapacity} converted the problem of finding feedback Gaussian channel capacity into a form of stochastic control and used dynamic programming to compute the n-block capacity.

\section{Feedback Control Interpretation of Feedback Capacity for Gaussian Channels }\label{sec.interpretation}
\indent We recall that {\cite{Elia2004} provides feasible communication schemes for a given channel. However, in order to construct capacity-achieving feedback codes, two things are necessary a) the capacity must be known and b) an unstable controller, in the terminology of this paper, must be found. Both steps are not easy, although we know from \cite{Kim10} that linear scheme is capacity achieving.

In this section, we propose a modification of the approach of  {\cite{Elia2004}  that provides a derivation of the feedback capacity formula for finite dimensional LTI Gaussian channels, from control theory principles. The proposed approach based on \cite{Elia2004} provides feasible feedback communication schemes with guaranteed transmission rate. Extensive simulations show that this coding scheme can achieve (arbitrarily) close to the feedback capacity.



%

\begin{figure}
\begin{center}
\includegraphics[scale=0.7]{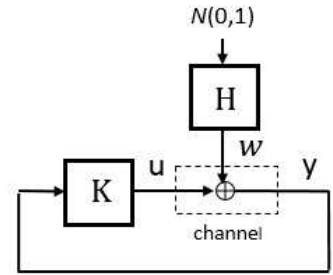}
\caption{Gaussian channels with feedback.}
\label{noisyfb_scheme}
\end{center}
\end{figure}
Given the channel/plant (see Fig. \ref{noisyfb_scheme})
\begin{equation}
y_i=u_i+w_i,
\label{channel.LTI}
\end{equation}
where $w_i$ satisfies Assumption \ref{LTIFD.ass}. As shown in Fig. \ref{noisyfb_scheme}, we are interested in the closed loop stabilization problem over the given channel. Note that the noiseless feedback communication scheme will measure $y$ and use it help produce $u$ with one step delay.

Any such closed loop system must  produce $u_i$ with the required transmission power, $P$. Any controller that produces bounded transmission power in the loop must be a stabilizing controller. Differently from  {\cite{Elia2004} where a specific plant was stabilized over the channel, here we consider the map from $u_i$ to $y_i$ as the plant, $\mathbb{G}$.
Since $\mathbb{G}$ is stable, all strictly causal finite dimensional LTI stabilizing controllers for such a plant have the following expression \cite{Doyle92}, represented as transfer functions:
\begin{equation}\label{youla.eq}
{\mathbb K}=-{\mathbb Q}(I+{\mathbb G}{\mathbb Q})^{-1}=-\mathbb{Q}(I+\mathbb{Q})^{-1}
\end{equation}
where $\mathbb{Q}$ is any finite-dimensional LTI strictly causal stable system, i.e.
$\mathbb{Q}\in \mathcal{RH}_2$ strictly causal.

The above parametrization is known as the Youla parametrization of stabilizing controllers. Working with ${\mathbb{Q}}$ instead of the set of stabilizing controllers ${\mathbb K}$ is more convenient. The main advantage comes from the fact that the above transformation convexifies the set of achievable closed loop maps by a stabilizing controller.
In particular, consider the complementary sensitivity map $\mathbb{T}$ from $w$ to $u$ and sensitivity map $\mathbb{S}$ from $w$ to $y$, respectively.
Simple feedback operations lead to
$$
\mathbb{T}(\mathbb{K})=\mathbb{K}(1+\mathbb{G}\mathbb{K})^{-1}$$
while
$$
\mathbb{S}(\mathbb{K})=(1+\mathbb{G}\mathbb{K})^{-1}
$$

Substituting (\ref{youla.eq}) into $\mathbb{S}$ and $\mathbb{T}$, it follows that
\begin{equation}
\begin{split}
\mathbb{S}(\mathbb{K})=&(I-\mathbb{G}{\mathbb Q}(I+{\mathbb G}{\mathbb Q})^{-1})^{-1}\\
=&I+{\mathbb G}{\mathbb Q}=I+\mathbb{Q}\\
\end{split}
\label{S_function}
\end{equation}
and
\begin{equation}
\begin{split}
\mathbb{T}(\mathbb{K})=&-{\mathbb Q}(I+{\mathbb G}{\mathbb Q})^{-1}S\\
=&-{\mathbb Q}(I+{\mathbb G}{\mathbb Q})^{-1}(I+{\mathbb G}{\mathbb Q})\\
=&-\mathbb{Q}\\
\end{split}
\label{T_function}
\end{equation}
It was shown in \cite{Elia2004} that the Bode Integral formula of $\mathbb{S}$, the sensitivity function is tightly connected to the Directed Information\footnote{Directed information, firstly defined by Massey \cite{Massey1990}, has been vastly used in characterizing the capacity of channels with feedback \cite{Yang_finteState}\cite{kim08} \cite{Tati09}\cite{chong_isit11_capacity}. Moreover, it has interpretation on portfolio theory, data compression and hypothesis testing \cite{Permuter11}} rate of the channel, which measures the reliable transmission rate through the channel, see also \cite{Kramer_thesis,Tati09}\footnote{As shown in \cite{Tati09}, the feedback capacity of arbitrary channels can be obtained by maximizing the rate of directed information of the closed-loop system}.
Specifically, the rate
\begin{equation}
\begin{split}
R=&\lim_{n\rightarrow \infty}\frac{1}{n}I(U^n \rightarrow Y^n)\\
& = \frac{1}{4\pi}\int_{-\pi}^\pi \log |\mathbb{S}(e^{i\theta})|^2d \theta\\
& = \frac{1}{4\pi}\int_{-\pi}^\pi \log |(1+\mathbb{G}\mathbb{K})^{-1}|^2d \theta\\
& = \frac{1}{4\pi}\int_{-\pi}^\pi \log |I+{\mathbb Q}(e^{i\theta})|^2d \theta\\
\end{split}
\label{equ.eigen.rate}
\end{equation}
where $I(U^n \rightarrow Y^n)$ denotes the directed information from random sequence $U^n = \lbrace U_i\rbrace_{i=1}^{n}$ to random sequence $Y^n = \lbrace Y_i\rbrace_{i=1}^{n}$. The detailed derivation of the above equalities can be found in Theorem 4.6 \cite{Elia2004}. According to Poisson-Jensen's formula, since $\mathbb{S}$ is stable, this integral only depends on the zeros of $\mathbb{S}$ outside the unit disc,  i.e., non-minimum phase (NMP) zeros. Given the relation between $\mathbb{S}$ and (\ref{youla.eq}), i.e., $\mathbb K = -{\mathbb Q}\mathbb{S}^{-1}$, this implies that $\mathbb{K}$ must be unstable if $\mathbb{Q}$ do not have NMP zeros that cancel those of $I+\mathbb{Q}$ (which is always true). This fact is later verified by numerical examples.

The average power of channel input $u$ required for stabilization, under the current  assumptions, is given by
\begin{equation}
\begin{split}
\lim_{n\rightarrow \infty} \frac{1}{n}\sum_{i=1}^{n}u_i = &\frac{1}{2\pi}\int_{-\pi}^\pi |{\mathbb S}_u(e^{i\theta})|^2d\theta\\
=&\frac{1}{2\pi}\int_{-\pi}^\pi |\mathbb{T}(e^{i\theta})\mathbb{H}(e^{i\theta})|^2d\theta\\
=&\frac{1}{2\pi}\int_{-\pi}^\pi |-\mathbb{Q}(e^{i\theta})\mathbb{H}(e^{i\theta})|^2d\theta\\
=&\frac{1}{2\pi}\int_{-\pi}^\pi |\mathbb{Q}(e^{i\theta})|^2{\mathbb S}_w(e^{i\theta})d\theta.\\
\end{split}
\end{equation}
Note that the first equality follows from Parseval's theorem. Thus, we can search over strictly causal $\mathbb{Q}\in \mathcal{RH}_2$ to maximize the Bode Integral Formula over the average power for stabilization constraint.
Therefore, the largest achievable rate of all strictly causal LTI stabilizing controllers is given by
\begin{equation}
\begin{split}
R_{max}=&\max_{\mathbb{Q}}\frac{1}{4\pi}\int_{-\pi}^{\pi}\log |1+\mathbb{Q}(e^{i\theta})|^2 d\theta,\\
s.t. \quad  &\frac{1}{2\pi}\int_{-\pi}^{\pi}|\mathbb{Q}(e^{i\theta})|^2\mathbb{S}_w(e^{i\theta})d\theta\leq P,\\
& \mathbb{Q}(e^{i\theta}) \quad \text{is strictly causal},\;{\mathbb{Q}\in {\mathcal{RH}_2}}.\\
\end{split}
\label{youla.cap.eq}
\end{equation}

Notice that the above optimization, derived from Youla parameterization, is identical to (\ref{capacity_short01}) derived from information theory.
In summary,
\begin{enumerate}
\item the above derivation extends the feedback control interpretation of communication system over Gaussian channels with access to feedback and shows how the Youla parameter $\mathbb{Q}$ is central to the feedback capacity problem;
\item as will be seen next, feasible feedback coding schemes can be explicitly constructed from controller ${\mathbb K}$ with guaranteed transmission rate (in the sense of Shannon);
\end{enumerate}

\begin{figure*}
\begin{center}
\includegraphics[scale=0.4]{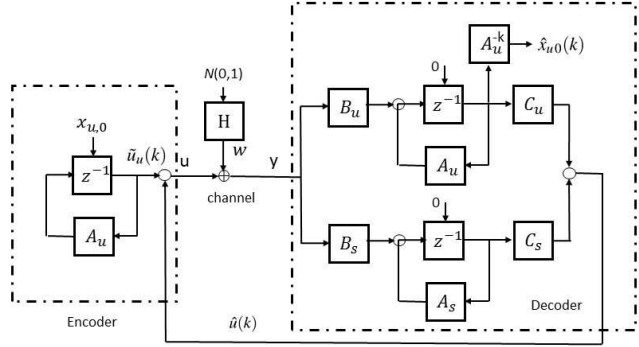}
\caption{Decomposition of controller $\mathbb{K}$ into feedback encoder and decoder.}
\label{fig:codingStructure}
\end{center}
\end{figure*}

\subsection{Feasible  Coding Schemes}
Once a feasible $\mathbb{Q}$ is found for the above optimization, (possibly optimal or arbitrarily close to optimal) we can obtain a system
$\mathbb{K}=\mathbb{Q}(I+\mathbb{Q})^{-1}$ stabilizing the channel within the prescribed input average power limit.
We show  how to construct a feasible feedback coding scheme, which is deterministic (time-invariant) and has double exponential decaying decoding error probability and computable transmission rate. We follow \cite{Elia2004}.

First of all, we present controller $\mathbb{K}$ as an LTI single-input-single-output (SISO) finite-dimensional discrete-time unstable system with the following state-space model:
\begin{equation}
\begin{split}
\mathbb{K}: \qquad \begin{bmatrix} x_s(k+1) \\ x_u(k+1)\end{bmatrix} &= \begin{bmatrix} A_s & 0 \\ 0 & A_u\end{bmatrix} \begin{bmatrix} x_s(k) \\ x_u(k)\end{bmatrix} + \begin{bmatrix} B_s \\ B_u\end{bmatrix}y(k)\\
u(k) &= \begin{bmatrix} C_s & C_u\end{bmatrix} \begin{bmatrix} x_s(k) \\ x_u(k)\end{bmatrix}\\
\end{split}
\label{codingScheme_SS}
\end{equation}
The eigenvalues of $A_u$ are outside the unit disc while the eigenvalues of $A_s$ are all strictly inside the unit disc (strictly stable). Without loss of generality we assume that $A_s$ and $A_u$ are in Jordan form. Assume $A_u$ has $m$ eigenvalues, denoted by $\lambda_i(A_u), i = 1,2,\cdots, m$ .

Starting with the decoder, we decompose ${\mathbb K}$ as follows. We present the simplest solution here, others are possible.\\
{\em Decoder}\\
The decoder runs ${\mathbb K}$ driven by $y$.
$$
\begin{array}{ccl}
x_s(k+1)&=&A_sx_s(k)+B_sy(k),\; x_{s}(0)=0\\
\hat{x}_u(k+1)&=&A_u\hat{x}_u(k)+B_uy(k),\; \hat{x}_{u}(0)=0
\end{array}
$$
It produces two signals:
an estimate of the initial condition of the encoder
$$
\hat{x}_{u\,0}(k)=A_u^{-k}\hat{x}_u(k).
$$
and a feedback signal
$$
\hat{u}(k) = \begin{bmatrix} C_s & C_u\end{bmatrix} \begin{bmatrix} x_s(k) \\ \hat{x}_u(k)\end{bmatrix}\\
$$
{\em Encoder}\\
The encoder runs the following dynamics
$$
\begin{array}{rcl}
\tilde{x}_u(k+1)&=&A_u\tilde{x}_u(k),\;\tilde{x}_u(0)=x_{u\,0},\\
\tilde{u}_u(k)&=&C_u\tilde{x}_u(k)
\end{array}
$$
it receives $\hat{u}$ and produces the channel input
$$
u(k)=\tilde{u}_u(k)+\hat{u}(k)
$$
Since the closed loop is stable,  $x_u(k)=\tilde{x}_u(k)+\hat{x}_u(k)$ goes to zero with time if the noise is not present. This implies that
$\hat{x}_u(k)\to -\tilde{x}_u(k)$. Thus,
$-\hat{x}_{u\,0}(k)$ is an  estimate at time $k$ of $\tilde{x}_u(0)=x_u(0)=x_{u,0}$. This coding scheme is illustrated in Fig.\ref{fig:codingStructure}.

Next theorem describes how fast messages associated with each $x_{u,0}$ is transferred to $-\hat{x}_{u\,0}(k)$
in the presence of the channel noise.
\begin{theorem}
Consider stationary Gaussian channels in (\ref{channel.LTI}). Given a filter $\mathbb{Q}(e^{i\theta})\in \mathcal{RH}_2$, the coding scheme described above based on the decomposition of $\mathbb{K} = -\mathbb{Q}(I+\mathbb{Q})^{-1}$ achieves a reliable transmission rate (in the sense of Shannon) at $\frac{1}{2\pi}\int_{-\pi}^{\pi}\log |1+\mathbb{Q}(e^{i\theta})|d\theta =\sum_{i=1}^{m} \log|\lambda_i(A_u)|$  \textit{bits/channel use} and has double exponential decaying error probability.
\label{thm_capacity_achieving_code}
\end{theorem}
The proof is omitted as it directly follows from (\ref{equ.eigen.rate}) and Theorem 4.3 in \cite{Elia2004}. We remark that the above transmission rate is achieved by allocating the message index $x_{u,0}$ at the centroid of small hypercubes generated by an unit hypercube in the coordinate system depending on $A_u$. We refer interested readers to Theorem 4.3 in \cite{Elia2004} for details.

In summary, the above discussion provides an approach to construct feasible feedback coding schemes over stationary finite dimensional Gaussian channels, by leveraging Youla parameter $\mathbb{Q}$. However, we need to solve (\ref{youla.cap.eq}) which is an infinite dimensional non-convex optimization problem. In the rest of the paper, we provide an approach to find an asymptotic capacity-achieving capacity upper bounds and the resulting filter $\mathbb{Q}$, based on which feedback codes can be constructed as described in this section.

\section{Upper Bounds on Capacity}
In this section, we firstly present an alternative characterization of Gaussian feedback capacity by leveraging the inverse Fourier transform. Based on this characterization, a sequence of asymptotic capacity-achieving upper bounds is proposed and evaluated by solving finite dimensional optimization problems.
\subsection{Alternative characterization of Gaussian Feedback Capacity}
\indent We focus on the optimization problem (\ref{youla.cap.eq}) or (\ref{capacity_short01}). In what follows, we characterize the Gaussian feedback capacity by imposing the causality constraints in terms of the inverse Fourier transform. 
\begin{lemma}\label{lemma_symmetric_filter}
Under Assumption \ref{LTIFD.ass},
there exists an optimal solution $\hat{\mathbb{Q}}(e^{i\theta})$ for (\ref{capacity_short01}) with $\hat{\mathbb{Q}}(e^{i\theta}) = \hat{\mathbb{Q}}^*(e^{-i\theta})$ where $\hat{\mathbb{Q}}^*$ indicates the complex conjugate. Furthermore, the feedback capacity can be characterized by
\begin{equation}
\begin{split}
C_{fb}=&\max_{\Gamma}\frac{1}{4\pi}\int_{-\pi}^{\pi}\log ((1+a(\theta))^2+ b(\theta)^2 )d\theta\\
s.t. \quad &\frac{1}{2\pi}\int_{-\pi}^{\pi}\left(a^2(\theta)+b^2(\theta)\right) S_w(\theta)d\theta\leq P,\\
&\text{and strictly causal filter frequency domain constraints}\\
&\int_{-\pi}^{\pi} a(\theta)\cos(n\theta) d\theta + \int_{-\pi}^{\pi} b(\theta) \sin(n\theta)d\theta = 0\\
& \quad n = 0, 1,2,\cdots,\infty \\
\end{split}
\label{formula_stationaryGuassian_upperbound_equi}
\end{equation}
where the maximum is taken over a functional set $\Gamma$ defined as
\begin{equation}
\begin{split}
\Gamma =& \lbrace a(\theta), b(\theta): [-\pi, \pi] \rightarrow \mathbb{R} |  a(\theta), b(\theta)\in \mathcal{L}_2\rbrace.\\
\end{split}
\label{def_set_Gamma}
\end{equation}
\label{lemma:symmetric_filter}
\end{lemma}
Here $S_w(\theta)$ refers to $\mathbb{S}_w(e^{i\theta})$ for simplicity.
The basic idea of this characterization is that the strict causality can be imposed on the non-positive index coefficients of the inverse Fourier transform of $\mathbb{Q}(e^{i\theta})$.

\subsection{Upper bounds}
We next obtain upper bounds to $C_{fb}$ by considering only a finite number of causality constraints. The h-upper-bound, denoted by $C_{fb}(h)$, is defined as
follows:
\begin{equation}\label{cfbh.eq}
\begin{split}
C_{fb}(h)=&\sup_{\Gamma}\frac{1}{4\pi}\int_{-\pi}^{\pi}\log ((1+a(\theta))^2+ b(\theta)^2 )d\theta\\
s.t. \quad &\frac{1}{2\pi}\int_{-\pi}^{\pi}\left(a^2(\theta)+b^2(\theta)\right) S_w(\theta)d\theta\leq P,\\
&\int_{-\pi}^{\pi} a(\theta)\cos(n\theta) d\theta + \int_{-\pi}^{\pi} b(\theta) \sin(n\theta)d\theta = 0\\
& \quad n = 0, 1,2,\cdots, h.\\
\end{split}
\end{equation}
Since $C_{fb}(h) \geq C_{fb}(h+1)$ for any $h\geq 0$, and the sequence $C_{fb}(h)$ is bounded from below, we have $C_{fb} = \lim_{h\rightarrow \infty} C_{fb}(h)$ being approached from above.

 Note $C_{fb}(h)$ is still a semi-infinite dimensional non convex optimization problem.
The next theorem shows that the optimal solution to
$C_{fb}(h)$ exits in $\mathcal{L}_2$. Moreover it characterizes the Lagrangian dual problem and shows that there is no duality gap between the primal and dual problems.
In what follows we focus on non flat channels as the feedback capacity of flat channels is solved.
\begin{theorem}(\textit{main result})\label{strongdual.thm}
Under Assumption \ref{LTIFD.ass}, further assume that $\mathbb{S}_w$ is non-flat (e.g. $\mathbb{S}_w(e^{i\theta})$ is not constant over $\theta$). Let
$$
\begin{array}{l}
A(\theta) =[\cos(\theta),\cos(2\theta), \cdots,\cos(h\theta)]',\\
B(\theta) =[\sin(\theta),\sin(2\theta), \cdots,\sin(h\theta)]'.
\end{array}
$$
For $\lambda\geq 0$, $\eta\in R^h$, and $\eta_0\in R$, define
$$
r^2(\theta) =(2\lambda S_{w}(\theta)+\eta'A(\theta)+\eta_0)^2+(\eta'B(\theta))^2.
$$
Then, \\
a) The Lagrangian dual of problem $C_{fb}(h)$ in  (\ref{cfbh.eq}) is given by the following optimization:
\begin{equation}\label{dual.eq}
(D): \mu_h = - \max_{\lambda\geq 0,\eta\in \mathbb{R}^{h},\eta_0\in R}g(\lambda,\eta,\eta_0)
\end{equation}
where
\begin{equation}\label{thm.dual.eq}
\begin{split}
&g(\lambda,\eta,\eta_0)\\
=&\displaystyle\frac{1}{2\pi}\int_{-\pi}^\pi \left[\frac{1}{2}\log(2\lambda S_{w}(\theta)-\nu(\theta))-\frac{r^2(\theta)}{2\nu(\theta)}+\lambda S_{w}(\theta)\right]d\theta\\
&-\lambda P+\eta_0+\frac{1}{2}.\\
\end{split}
\end{equation}
with $
\nu(\theta)=\frac{-r^2(\theta)+\sqrt{r^4(\theta)+8\lambda S_{w}(\theta)r^2(\theta)}}{2}$.\\
b) (D) is equivalent to the following convex optimization problem
\begin{equation}\label{dual2.eq}
\mu_h = - \max_{\begin{array}{l}\lambda\geq 0,\eta\in \mathbb{R}^{h},\eta_0\in R\\ \nu(\theta)\geq 0 \in C^\infty_{[-\pi,\pi]}\end{array}}
\tilde{g}(\lambda,\eta,\eta_0,\nu(\theta))
\end{equation}
where
\begin{equation}\label{thm.dual2.eq}
\begin{split}
&\tilde{g}(\lambda,\eta,\eta_0,\nu(\theta))\\
=&\displaystyle\frac{1}{2\pi}\int_{-\pi}^\pi \left[\frac{1}{2}\log(2\lambda S_{w}(\theta)-\nu(\theta))-\frac{r^2(\theta)}{2\nu(\theta)}+\lambda S_{w}(\theta)\right]d\theta\\
&-\lambda P+\eta_0+\frac{1}{2}.\\
\end{split}
\end{equation}
c) Furthermore, $C_{fb}(h)=\mu_h$, and the optimal $\mathbb{Q}_h(e^{i\theta}) = a(\theta)+ib(\theta)$ for $C_{fb}(h)$ is obtained from the optimal solution of (\ref{dual.eq}) or (\ref{dual2.eq}) as follows:
\begin{eqnarray}\label{thm_xy_dual}
a(\theta)&=&\frac{2\lambda S_{w}(\theta)+\eta'A(\theta)+\eta_0}{\nu(\theta)}-1 \quad a.e.\\
b(\theta)&=&\frac{\eta'B(\theta)}{\nu(\theta)}  \quad a.e.
\end{eqnarray}
\label{lemma.strong.dual}
\end{theorem}
%
\subsection{Computing $C_{fb}(h)$}
\indent Although the dual problem of $C_{fb}(h)$ can be casted into a convex
optimization (\ref{dual.eq}) with finite number of variables. The problem it is not easily computable since the cost is an  integral, not explicitly computable in terms of the variables.
A natural practical approach would be to approximate the integral with a finite sum by discretizing $\theta$. We apply such discretization to (\ref{dual2.eq}) (with spacing $\frac{\pi}{m}$) and introduce the following finite dimensional convex problem from (D). Given $m$,  consider
\begin{equation}{\label{opt_upperbound_approximate}}
C_{fb}(m,h)=-\max_{\lambda\geq 0,\eta,\eta_0, \nu_i}g_m(\lambda,\eta,\eta_0,\nu_i)
\end{equation}
where
\begin{equation*}
\begin{split}
&g_m(\lambda,\eta,\eta_0,\nu_i)\\
=& \frac{1}{2m}\sum_{i=1}^{2m}\frac{1}{2}\log(2\lambda S_{w}(\theta_i)-\nu_i)+\frac{1}{2}-\frac{r^2(\theta_i)}{2\nu_i}+\lambda S_{w}(\theta_i)\\
&-\lambda P+\eta_0,\\
\end{split}
\end{equation*}
and $\theta_i = -\pi+ \frac{\pi}{m}(i-1)$. \\
Clearly, $C_{fb}(m,h)$ is only an approximation of $C_{fb}(h)$, although the approximation gets arbitrarily close to $C_{fb}(h)$ and its solution as $m\to \infty$;  since the integration is over a compact set, and the variables are either finite dimensional or continuous.

Notice that the optimization (\ref{opt_upperbound_approximate}) is in a convex form.
The $\log$ of an affine function is concave. $\frac{r^2(\theta_i)}{\nu_i}$ is a quadratic (composed with an affine function of the variables)  over linear function, therefore convex.  Thus,  (\ref{opt_upperbound_approximate}) can be efficiently  solved  with standard convex optimization tools, e.g. CVX.


Based on the solution to (\ref{opt_upperbound_approximate}), we can actually obtain a guaranteed upper bound on $C_{fb}(h)$ for each $m$ using the upper bound property of dual feasible solutions.
Let $\lambda ^m, \eta^m, \eta_0^m, \nu_i^m$ be the optimal solution to  (\ref{opt_upperbound_approximate}). Here we are omitting the dependence on $h$ for simplicity. Let
\begin{equation}{\label{C_h_m}}
\overline{C_{fb}(m,h)} = - g(\lambda^m,\eta^m,\eta_0^m)
\end{equation}
where $g(\cdot)$ is defined in (\ref{thm.dual.eq}). Clearly,  $\overline{C_{fb}(m,h)}$ is easily computable to  arbitrary accuracy.



\begin{corollary}
Given $h\geq 0$, $\overline{C_{fb}(m,h)} \geq C_{fb}(h) \geq C_{fb}$ for $\forall m > 0 $ and $\lim_{m\rightarrow\infty}\overline{C_{fb}(m,h)} = C_{fb}(h)$.
\end{corollary}
\begin{proof}
From the solution of (\ref{opt_upperbound_approximate}),  $\lambda ^m, \eta^m, \eta_0^m$, $\nu_i^m$, we know that  $\lambda ^m, \eta^m, \eta_0^m$ are feasible  for (\ref{dual.eq}). However,  any feasible dual solution provides a cost $- g(\lambda^m,\eta^m,\eta_0^m)$ which is an upper bound on $C_{fb}(h)$.
\end{proof}
Extensive numerical simulations show that for sufficiently large $m$ the upper bound $\overline{C_{fb}(m,h)}$, and therefore $C_{fb}(h)$  are close to the capacity even  for small $h$.

\section{Lower Bounds on Capacity}
In the previous section we have introduced a finite dimensional convex optimization, (\ref{opt_upperbound_approximate}). From its optimal solutions we were able to obtain convergent upper bounds on $C_{fb}$. In this section, we show that from the solution to (\ref{opt_upperbound_approximate}) we can obtain lower bounds on $C_{fb}$.

The primal problem associated with  (\ref{opt_upperbound_approximate}) is the following optimization. This can be verified analogously to Theorem \ref{strongdual.thm}}.
\begin{equation}\label{cfbhfd.eq}
\begin{split}
C_{fb}(m,h)=&\min_{a_i\in R,b_i\in R}\frac{1}{4m}\sum_{i=1}^{2m}\log ((1+a_i)^2+ b_i^2 )\\
s.t. \quad &\frac{1}{2m}\sum_{i=1}^{2m}\left(a^2_i+b^2_i\right) S_w(\theta_i) \leq P,\\
&\sum_{i=1}^{2m}a_i\cos(n\theta_i) + \sum_{i=1}^{2m} b_i \sin(n\theta_i) = 0\\
& \quad n = 0, 1,2,\cdots, h.\\
&\theta_i = -\pi+ \frac{\pi}{m}(i-1)
\end{split}
\end{equation}
This problem is non convex, however, from the solution to (\ref{opt_upperbound_approximate}),  we can construct
\begin{eqnarray}\label{thm_xy_dual_approx}
a_i&=&\frac{2\lambda S_{w}(\theta_i)+\eta'A(\theta_i)+\eta_0}{\nu_i}-1 \label{a.eq}\\
b_i&=&\frac{\eta'B(\theta_i)}{\nu_i}\label{b.eq}
\end{eqnarray}
which are feasible for (\ref{cfbhfd.eq}).


Note that (\ref{a.eq}) and (\ref{b.eq}) can be interpreted as the coefficients of a sampled spectrum $\tilde{\mathbb{Q}}_{m,h}$.  As such, they are associated with a periodic impulse response signal $\tilde{q}_{m,h}(k)=\displaystyle\sum_{n=-\infty}^\infty c_n \delta(k-n)$, with fundamental period $2m$, where
\begin{equation}
\begin{array}{l}
c_n =\displaystyle\frac{1}{2m}\sum_{i=1}^{2m}a_i\cos(n\theta_i)-b_i\sin(n\theta_i)\\
\theta_i = -\pi+ \frac{\pi}{m}(i-1).
\end{array}
\label{alg:L2_filter}
\end{equation}
Consider one $2m$-period truncation of $\tilde{q}_{m,h}(k)$, denoted by $q_{m,h}(k)$. We next perform a causal projection on $q_{m,h}(k)$  by zeroing all the non strictly causal coefficients. Finally, we extend the signal defined from $-m+1$ to $m$ to have zero value outside the interval $[-m+1,m]$. Let $q^c_{m,h}(k)$  denote such strictly causal Finite Impulse Response signal and let its Fourier Transform be $\mathbb{Q}^c_{m,h}$.
$\mathbb{Q}^c_{m,h}$ will satisfy the power constraint
$$
\frac{1}{2\pi}\int_{-\pi}^{\pi}|\mathbb{Q}_{m,h}^c(e^{i\theta})|^2\mathbb{S}_w(e^{i\theta})d\theta\leq P
$$
for $m$ large enough.  However  if for some $m$ it does not, then we can scale  $\mathbb{Q}^c_{m,h}$ appropriately so that the scaled  $\mathbb{Q}_{m,h}^c(e^{i\theta})$ does.
The end result of this procedure is that $\mathbb{Q}_{m,h}^c(e^{i\theta})$ is feasible for
$C_{fb}$ in (\ref{capacity_short01})(\ref{youla.cap.eq})(\ref{formula_stationaryGuassian_upperbound_equi}).

Specifically:
\begin{enumerate}
\item \textit{Filter Construction}: given $m,h>0$, solve (\ref{opt_upperbound_approximate}) and obtain solution $(\lambda_{m,h}^*, \eta_{m,h}^*,\eta_{0,m,h}^*)$. Obtain  $a_i,b_i$, $i=1,\ldots,2m$ from
(\ref{thm_xy_dual_approx}) with $(\lambda_{m,h}^*, \eta_{m,h}^*,\eta_{0,m,h}^*)$.
\item \textit{One Period Truncation and Causal Projection}: compute the strictly casual part of one period of the impulse response by computing the coefficients $c_n$ $(n=1,\ldots m)$ by
\begin{equation}
\begin{array}{l}
c_n =\displaystyle\frac{1}{2m}\sum_{i=1}^{2m}a_i\cos(n\theta_i)-b_i\sin(n\theta_i)\\
\theta_i = -\pi+ \frac{\pi}{m}(i-1).
\end{array}
\label{alg:L2_filter}
\end{equation}
Then, construct a strictly causal filter $\mathbb{Q}_{m,h}^{c}(z)=\displaystyle\sum_{n=1}^m c_n z^{-n}$.
Lower order rational approximations may be obtained by Hankel model reduction methods, e.g. \cite{Kung_svd78}, if desired.

\item \textit{Power Scale}: Let $p:=\frac{1}{2\pi}\int_{-\pi}^{\pi}|\mathbb{Q}^{c}_{h,m}(e^{i\theta})|^2\mathbb{S}_w(e^{i\theta})d\theta$. $p$ can be computed by computing the ${\cal H}_2$ norm squared for $\mathbb{Q}^c_{h,m}(z)H(z)$. This can be done in state-space by computing the solution to a Lyapunov equation.  We next rescale $\mathbb{Q}^c_{h,m}$ by $\alpha_{m,h}=\sqrt{P/p}$ and obtain $\mathbb{Q}_{m,h}^{cp}=\alpha_{m,h}\mathbb{Q}_{m,h}^{c}$ $(\alpha_{m,h}>0)$,  such that the power budget is satisfied.
\item \textit{Coding Scheme Construction}: $\mathbb{Q}_{m,h}^{cp}$  is a feasible solution to (\ref{formula_stationaryGuassian_upperbound_equi}). Thus  we can construct a feedback coding scheme from $\mathbb{K} = -\mathbb{Q}_{m,h}^{cp}(I+\mathbb{Q}_{m,h}^{cp})^{-1}$ (and its corresponding state-space representation).
\end{enumerate}
Theorem \ref{thm_capacity_achieving_code} shows this coding scheme $\mathbb{K}$ achieves a rate $R(h,m)= \frac{1}{2\pi}\int_{-\pi}^{\pi}\log |1+{\mathbb{Q}}_{h,m}^{cp}(e^{i\theta})|d\theta$ which is a lower bound on the capacity, and has double exponential decaying error probability.

Combined with the upper bound $\overline{C_{fb}(h,m)}$, a numerical capacity gap can be evaluated by $\overline{C_{fb}(h,m)} - R(h,m)$.
\section{Examples}
In summary, we can compute the upper bounds $C(h)$ with desired accuracy by solving the finite dimensional convex optimization (\ref{opt_upperbound_approximate}) with sufficiently large $m$.  We then construct a strictly causal filter $\mathbb{Q}$ for optimization (\ref{youla.cap.eq}) from the solution of (\ref{opt_upperbound_approximate}). From $\mathbb{Q}$ we obtain $\mathbb{K}$ and a feasible feedback communication scheme as described in Section \ref{sec.interpretation}.A.
The transmission rate can be computed from the NMP zeros of the sensitivity function $S=I+\mathbb{Q}$ as in Theorem \ref{thm_capacity_achieving_code}.
We next present some examples.
\subsection{Examples}
\begin{example}
Consider a first-order moving average (i.e. MV(1)) Gaussian process $W_i = U_i + 0.1U_{i-1}$ where $U_i$ is a white Gaussian process with zero mean and unit variance. The power spectral density is $\mathbb{S}_w(e^{i\theta}) = |1+ 0.1 e^{-i\theta}|^2$. Given power constraint $P = 10$, our proposed upper
and lower bounds converge to $1.7688$ \textit{bits/channel use}, being consistent with that computed from the closed form solution (44) in \cite{Kim10}.
\end{example}
\begin{example}
Consider the following second order moving average Gaussian process with
$$
W(z)=1+0.1z^{-1}+0.5z^{-2}
$$
with associated
$$
S_w(z)=|W(e^{i\theta})|^2.
$$
While neither the value of capacity or the optimal codes is known for this generalized Gaussian noise, both of them can be efficiently obtained from our approach.
With power constraint $P = 10$, the capacity is evaluated as $1.9194$ \textit{bits/channel use} (rounded to $4$ decimals).
This value is obtained with $h=6$ and $m=40$.
See Table \ref{table:convergence} for the convergence of upper and lower bounds as $h$ increases.
It is shown that the gap is vanishing quickly and the capacity is evaluated with high accuracy.

The optimal coding scheme $\mathbb{K}$  after order reduction via Hankel Singular Value Decomposition (H-SVD) on the finite impulse response is given by
$$
{\mathbb K}=-\frac{0.22026 (z+13.84) z^2}{(z^2 + 0.01755z + 0.03498) (z^2 + 0.4115z + 3.783)},
$$
which is unstable as expected. Note however that it has two complex conjugate unstable poles at
$$
p_{1,2}=-0.2057 \pm i 1.9340
$$
which would not be easy to find using the approach of \cite{Elia2004}. Also it can be verified that the achievable rate of $\mathbb{K}$ is $\log(|p_1||p_2|)=1.9194$ \textit{bits/channel use}.

The corresponding optimal closed loop Sensitivity function is
$$
I+{\mathbb Q}=\frac{(z^2 + 0.01755z + 0.03498) (z^2 + 0.4115z + 3.783)}{(z^2 + 0.1088z + 0.2644) (z^2 + 0.1z + 0.5)}
$$
As expected, the Sensitivity has the corresponding non-minimum phase zeros at the location of the unstable poles of ${\mathbb K}$. Note also that the optimal closed loop system includes dynamics that (partially) cancel the noise dynamics. The term $( z^2 + 0.1z + 0.5)$ is the numerator of $W(z)$. This feature is to be consistent with other examples where channels are modeled as minimum phase finite impulse response (FIR) filters.
Finally, Figure \ref{impulse.fig} shows the optimal impulse response of $\mathbb{Q}$, which is strictly causal as required.

\begin{figure}
\begin{center}
\includegraphics[scale=.5]{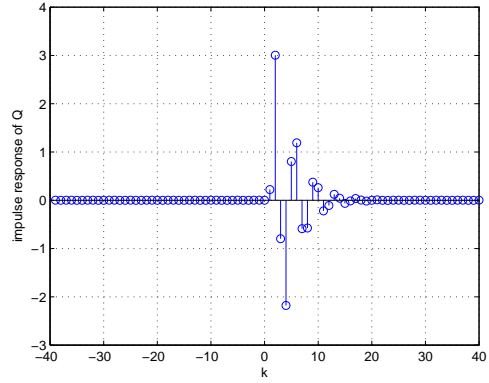}
\caption{Impulse response of ${\mathbb Q}$. }
\label{impulse.fig}
\end{center}
\end{figure}
\end{example}


\begin{figure*}
\begin{center}
\begin{tabular}{|c|c|c|c|}
\hline
& upper bound& lower bound & gap \\ \hline
h = 1 & 1.953615794213734 & 1.837997383645331& 0.115618410568404\\ \hline
h = 2 &1.919419110833023 &1.919133474756371& 2.856360766521071 $\times 10^{-4}$ \\ \hline
h = 3 &1.919395054344304 & 1.919215947145071& 1.791071992334192 $\times 10^{-4}$ \\ \hline
h = 4 &1.919358863350398 &1.919358573743238& 2.896071606972583$\times 10^{-7}$ \\ \hline
h = 5 &1.919358787261653 &1.919358689375164& 9.788648980268988$\times 10^{-8}$ \\ \hline
h = 6 & 1.919358744798872&1.919358744265310& 5.335623054492089$\times 10^{-10}$ \\ \hline
\end{tabular}
\caption{Convergence of upper and lower bounds. }
\label{table:convergence}
\end{center}
\end{figure*}

\section{Conclusion}
This paper studied the problem of computing the feedback capacity of stationary finite dimensional Gaussian channels. Firstly, the interpretation of feedback communication as feedback control over Gaussian channels was extended by leveraging \textit{Youla} parameterization. This new interpretation provides an approach to construct feasible feedback coding schemes with double exponentially decaying error probability. We next derived an asymptotic capacity-achieving upper bounds, which can be numerically computed by solving finite dimensional optimization. From the resulting filters that achieve upper bounds, feasible feedback coding schemes were constructed.
The convergence upper bound and the achievable lower bound allow to evaluate the closeness of the achievable rate to the feedback capacity.
From extensive examples, the sequence of lower bounds converges arbitrarily close to the capacity, yielding an asymptotic optimal feedback coding scheme.
We leave the lower bound convergence proof to further investigations.

\bibliographystyle{IEEEtran}
\bibliography{ref}

\newpage

\end{document}